\documentclass[11pt]{article}

\usepackage[margin=1in]{geometry}

\DeclareMathAlphabet{\mathbbold}{U}{bbold}{m}{n}
\usepackage{xspace,amsmath,amsfonts,amssymb,amsthm}
\usepackage{mathtools}
\usepackage[usenames,dvipsnames,svgnames,table]{xcolor}
\usepackage{thm-restate}

\usepackage[pagebackref]{hyperref}
\hypersetup{
    pdftitle={}, 
    pdfauthor={}, 
    colorlinks=true, 
    linkcolor=blue, 
    citecolor=blue, 
    urlcolor=blue 
}

\renewcommand{\backref}[1]{}

\renewcommand{\backrefalt}[4]{%
\ifcase #1 %
\or
[p.\ #2]%
\else
[pp.\ #2]%
\fi}

\usepackage{tikz,tikz-qtree}
\usepackage{wrapfig}
\usepackage{enumitem}

\usepackage{tocloft}

\makeatletter
\newcommand*\rel@kern[1]{\kern#1\dimexpr\macc@kerna}
\newcommand*\widebar[1]{%
  \begingroup
  \def\mathaccent##1##2{%
    \rel@kern{0.8}%
    \overline{\rel@kern{-0.8}\macc@nucleus\rel@kern{0.2}}%
    \rel@kern{-0.2}%
  }%
  \macc@depth\@ne
  \let\math@bgroup\@empty \let\math@egroup\macc@set@skewchar
  \mathsurround\z@ \frozen@everymath{\mathgroup\macc@group\relax}%
  \macc@set@skewchar\relax
  \let\mathaccentV\macc@nested@a
  \macc@nested@a\relax111{#1}%
  \endgroup
}
\makeatother

\usepackage{microtype}

\let\oldproofname=\proofname
\renewcommand{\proofname}
    {\upshape\bfseries{\oldproofname}}

\newtheorem{theorem}{Theorem}
\newtheorem{definition}[theorem]{Definition}

\newtheorem{lemma}[theorem]{Lemma}
\newtheorem{conjecture}[theorem]{Conjecture}
\newtheorem{claim}[theorem]{Claim}

\newtheorem{innerthmnum}{Theorem}
  {\renewcommand\theinnerthmnum{#1}\innerthmnum}
  {\endinnerthmnum}

\newcommand{\eq}[1]{\hyperref[eq:#1]{(\ref*{eq:#1})}}
\renewcommand{\sec}[1]{\hyperref[sec:#1]
    {Section~\ref*{sec:#1}}}
\newcommand{\thm}[1]{\hyperref[thm:#1]
    {Theorem~\ref*{thm:#1}}}
\newcommand{\lem}[1]{\hyperref[lem:#1]{Lemma~\ref*{lem:#1}}}
\newcommand{\clm}[1]{\hyperref[claim:#1]{Claim~\ref*{claim:#1}}}

\newcommand{\prop}[1]{\hyperref[prop:#1]
    {Proposition~\ref*{prop:#1}}}
\newcommand{\cor}[1]{\hyperref[cor:#1]
    {Corollary~\ref*{cor:#1}}}
\newcommand{\fig}[1]{\hyperref[fig:#1]{Figure~\ref*{fig:#1}}}
\newcommand{\tab}[1]{\hyperref[tab:#1]{Table~\ref*{tab:#1}}}
\newcommand{\alg}[1]{\hyperref[alg:#1]
    {Algorithm~\ref*{alg:#1}}}
\newcommand{\app}[1]{\hyperref[app:#1]
    {Appendix~\ref*{app:#1}}}
\newcommand{\conj}[1]{\hyperref[conj:#1]
    {Conjecture~\ref*{conj:#1}}}

\newcommand{\Goos}{{G\"o\"os}\xspace}

\newcommand{\B}{\{0,1\}}
\newcommand{\OR}{\textsc{OR}}

\newcommand{\BSum}{\textsc{BSum}}
\newcommand{\BKK}{\textsc{BKK}}
\newcommand{\CS}{\textsc{CS}}
\newcommand{\tOmega}{\tilde{\Omega}}
\newcommand{\tO}{\tilde{O}}
\DeclareMathOperator{\D}{D}
\DeclareMathOperator{\R}{R}
\DeclareMathOperator{\Q}{Q}
\DeclareMathOperator{\C}{C}
\DeclareMathOperator{\UC}{UC}
\DeclareMathOperator{\RC}{RC}

\DeclareMathOperator{\bs}{bs}
\DeclareMathOperator{\ts}{ts}
\DeclareMathOperator{\s}{s}
\DeclareMathOperator{\supp}{supp}

\DeclareMathOperator{\avdeg}{avdeg}
\DeclareMathOperator{\adeg}{\widetilde{\deg}}
\DeclareMathOperator{\aavdeg}{\widetilde{\avdeg}}

\DeclareMathOperator{\poly}{poly}

\newcommand{\UCmin}{\UC_{\min}}
\newcommand{\degp}{\deg^{+}}
\newcommand{\degpmin}{\deg^{+}_{\min}}
\newcommand{\adegp}[1]{\adeg^{+#1}}
\newcommand{\adegpmin}[1]{\adeg^{+#1}_{\min}}
\newcommand{\avdegp}{\avdeg^{+}}
\newcommand{\avdegpmin}{\avdeg^{+}_{\min}}

\newcommand{\aavdegpmin}[1]{\aavdeg^{+#1}_{\min}}

\def\RR{{\mathbb{R}}}

\def\N{{\mathbb{N}}}

\def\Q{{\mathrm{Q}}}

\renewcommand{\Pr}{\mathop{\bf Pr\/}}

\def\poly{{\mathrm{poly}}}
\def\OR{{\mathrm{OR}}}

\def\({\left(}
\def\){\right)}

\newcommand{\remove}[1]{}

\newcommand{\mathify}[1]{\ifmmode{#1}\else\mbox{$#1$}\fi}

\newcommand{\complexityclass}[1]{{\bf{#1}}\xspace}

\newcommand{\UP}{\complexityclass{UP}}
\newcommand{\CONP}{\complexityclass{coNP}}

\newcommand{\CC}{\mathrm{C}}

\newcommand{\UPcc}{\UP^{cc}}
\newcommand{\coNPcc}{\CONP^{cc}}

\def\B{{\{0,1\}}}

\newcommand{\bool}{\{0,1\}}

\title{Low-Sensitivity Functions from Unambiguous Certificates}
\author{Shalev Ben-David\thanks{
Partially supported by NSF.}\\
\small University of Maryland\\
\texttt{shalev@umd.edu} \and Pooya Hatami\thanks{Partially supported by the National Science Foundation under agreement No. CCF-1412958. 
}\\ \small DIMACS \& IAS \\  \texttt{pooyahat@math.ias.edu} 
\and Avishay Tal\thanks{Supported by the Simons Collaboration on Algorithms and Geometry, and by the National Science Foundation grant No. CCF-1412958.}
\\ \small IAS \\  \texttt{avishay.tal@gmail.com} 
}
\date{}

\begin{document}

\maketitle

\begin{abstract}
    We provide new query complexity separations against
    sensitivity for total Boolean functions:
    a power $3$ separation between deterministic
    (and even randomized or quantum)
    query complexity and sensitivity,
    and a power $2.22$ separation between certificate
    complexity and sensitivity.
    We get these separations by using a new connection between
    sensitivity and a seemingly unrelated measure called
    one-sided unambiguous certificate complexity ($\UCmin$).
    We also show that $\UCmin$ is lower-bounded by
    fractional block sensitivity,
    which means we cannot use these techniques to get
    a super-quadratic separation between $\bs(f)$ and $\s(f)$. We also provide a quadratic separation between the tree-sensitivity and decision tree complexity of Boolean functions, disproving a conjecture of \cite{GSTW16}.
    
 Along the way,
 we give a power $1.22$ separation between certificate complexity and one-sided unambiguous certificate complexity, improving the power $1.128$ separation due to \Goos~\cite{Goos15}. As a consequence, we obtain an improved $\Omega(\log^{1.22} n)$ lower-bound on the co-nondeterministic communication complexity of the Clique vs. Independent Set problem.
\end{abstract}
\thispagestyle{empty}
\newpage
\pagenumbering{arabic}

\section{Introduction}

Sensitivity is one of the simplest complexity measures
of a Boolean function. For $f:\B^n\to\B$ and $x\in\B^n$,
the sensitivity of $x$ is the number of bits of $x$
that, when flipped, change the value of $f(x)$.
The sensitivity of $f$, denoted $\s(f)$, is the maximum
sensitivity of any input $x$ to $f$. Sensitivity
lower bounds other important measures in
query complexity, such as deterministic query complexity
$\D(f)$, randomized query complexity $\R(f)$,
certificate complexity $\C(f)$, and block sensitivity $\bs(f)$
(see \sec{prelim} for definitions). $\sqrt{\s(f)}$ is
 a lower bound on quantum query complexity $\Q(f)$.

Despite its simplicity, sensitivity has remained mysterious.
The other measures are polynomially related to each other:
we have $\bs(f)\leq\C(f)\leq\D(f)\leq\bs(f)^3$ and
$\Q(f)\leq\R(f)\leq\D(f)\leq\Q(f)^6$. In contrast,
no polynomial relationship connecting
sensitivity to these measures is known, despite much interest (this problem was first posed by \cite{Nis91}. For a survey, see
\cite{HKP11}. For recent progress, see
\cite{AS11,Bop12,AP14,ABG+14,APV15,AV15,GKS15,Sze15,GNS+16,GSTW16,Tal16}).

Until recently, the best known separation between
sensitivity and any of these other measures was quadratic.
Tal \cite{Tal16} showed a powerd $2.11$ separation between
$\D(f)$ and $\s(f)$. In this work, we improve this
to a power $3$ separation, and also show functions for which
$\Q(f)=\tOmega(\s(f)^3)$ and $\C(f)=\tOmega(\s(f)^{2.22})$.

We do this by exploiting a new connection
between sensitivity and a measure called one-sided
unambiguous certificate complexity,
which we denote by $\UCmin(f)$. This measure, and particularly
its two-sided version $\UC(f)$
(which is sometimes called subcube complexity),
has received significant attention in previous
work (e.g.\ \cite{BOH90,FKW02,Sav02,Bel06,
KRS15,GPW15,Goos15,GJPW15,CKLS16,AKK16}),
in part because it corresponds
to partition number in communication complexity.
Intuitively,
$\UCmin(f)$ is similar to (one-sided) certificate complexity,
except that the certificates are required to be
\emph{unambiguous}:
each input must be consistent with only one certificate.
For a formal definition, see \sec{uc}.

We prove the following theorem.

\begin{restatable}{theorem}{main}\label{thm:main}
For any $\alpha\in\RR^+$,
if there is a family of functions with
$\D(f)=\tOmega(\UCmin(f)^{1+\alpha})$,
then there is a family of functions with
$\D(f)=\tOmega(\s(f)^{2+\alpha})$.
The same is true if we replace $\D(f)$ by
$\bs(f),\RC(f), \C(f),\R(f),\Q(f),$ and many other measures.
\end{restatable}

\thm{main} can be generalized from sensitivity $\s(f)$
to bounded-size block sensitivity $\bs_{(k)}(f)$
(block sensitivity where each block is restricted to have
size at most $k$). However, there is a constant factor loss
that depends on $k$.

We observe that cheat sheet functions
(as defined in \cite{ABK15}) have low
$\UCmin$; in particular, one of the functions in
\cite{ABK15} already has a quadratic separation between
$\Q(f)$ and $\UCmin(f)$, giving a cubic separation
between $\Q(f)$ and $\s(f)$. 

\begin{restatable}{corollary}{cormaina}\label{cor:main1}
There is a family of functions with $\Q(f)=\tOmega(\s(f)^3)$.
\end{restatable}

We observe that combining an intermediate step of our proof of Theorem~\ref{thm:main} and the construction of \cite{ABK15} gives a quadratic separation between tree sensitivity ($\ts(f)$) and $\Q(f)$ disproving a conjecture of \cite{GSTW16}. We defer the definition of tree sensitivity to Section~\ref{sec:s-bs-ts}. 

\begin{restatable}{corollary}{cortree}\label{cortree}
There is a family of functions with $\Q(f)=\tOmega(\ts(f)^2)$. 
\end{restatable}

To separate $\C(f)$ from $\s(f)$, we will use a function $f$ with a significant gap between $\C(f)$ and $\UCmin(f)$. \Goos~\cite{Goos15}, as part of the proof of his exciting $\omega(\log n)$ lower-bound for communication complexity of clique versus independent set problem, gave a construction of a function $f$ such that $\C(f) \geq \UCmin(f)^\alpha$ for $\alpha\approx1.128$.  
Using \Goos's function \cite{Goos15} would give a family of functions with 
$\C(f)=\Omega(\s(f)^{2.128})$. We show that it is possible to obtain an even better separation (\thm{UC1vsC-new} below), leading to the following separation between $\C(f)$ and $\s(f)$.

\begin{restatable}{corollary}{cormainb}\label{cor:main2}
There is a family of functions with $\C(f)=\Omega(\s(f)^{2.22})$.
\end{restatable}

\paragraph{New separation between $\C$ and $\UCmin$.} It is known that $\C(f)\leq \UCmin(f)^2$ (e.g., \cite{Goos15}), and  analogously in the communication complexity world $\coNPcc(f) \le \UPcc(f)^2$ (\cite{Yan91}). { Next, we discuss a polynomial separation between $\C$ and $\UCmin$ due to \cite{Goos15} that uses function composition.}

Throughout the years, Boolean function composition was used  extensively to separate different complexity measures; a non-exhaustive list includes
\cite{Aar08,Ambainis06,BDK16,GSS16,KT16,NisanSzegedy94,OT13,OWZST14,WZ88,SaksW86,Sherstov13, Tal13, Tal16}. The natural idea is to exhibit some constant separation between any two measures: $M(f)$ and $N(f)$ (i.e., $M(f) < N(f)$ for a constant size function $f$) and then to prove that $M(f^{k}) \le M(f)^k$ and $N(f^k) \ge N(f)^k$, for any $k\in \N$. This yields an infinite family of functions with polynomial separation between $M$ and $N$, as $N(f^k) \geq M(f^k)^{\log N(f)/\log M(f)}$. However, this approach does not work straightforwardly in an attempt to separate $\UCmin$ from $\CC$, since it is not necessarily true that $\UCmin(f^k)\leq \UCmin(f)^k$. \cite{Goos15} overcomes this barrier by considering gadgets over a larger alphabet where the letters of the alphabet are weighted. He constructs such a gadget using projective planes, and further shows how to compose gadgets over a weighted alphabet in a way that behaves multiplicatively for both $\UCmin$ and $\CC$. Finally, he shows how to simulate the weights and the larger alphabet with a Boolean function. The gadget $f_k$ constructed by \Goos satisfies $\CC(f_k) = k^2-k+1$ and $\UCmin(f_k)=\frac{k (k+1)}{2}$, whenever $k-1$ is a prime power. The optimum separation is obtained when $k=8$, giving a separation exponent of $\log(57)/\log(36)\geq 1.128$.

Since $\CC(f_k) \approx k^2$ and $\UCmin(f_k) \approx k^2/2$ and the separation exponent is $$\log(\CC(f_k))/\log(\UCmin(f_k)) \approx \log(k^2)/\log(k^2/2),$$ it seems that one should try to take $k$ as small as possible. However, the additive terms affect smaller $k$'s more significantly, making the optimum attained at $k=8$. This motivated us to try and reduce the weights in other ways, in order to improve the exponent. To do so, we introduce {\em fractional weights}. The argument of \Goos as is does not allow fractional weights, and in particular when Booleanizing the function, it seems inherent to use integer weights. We overcome this difficulty by considering fractional weights in intermediate steps of the construction, and then round them up at the end to get integral weights. 
We obtain the following separation.

\begin{restatable}[$\UCmin(f)$ vs $\C(f)$ - Improved]{theorem}{thm:UC1vsC-new} \label{thm:UC1vsC-new}
There exists an infinite family of Boolean functions $f_n: \B^n \to \B$ such that
	$\C(f_n) \geq \tilde{\Omega}\left( \UCmin(f_n)^{\frac{\log(38/3)}{\log(8)}}\right) \geq \Omega(\UCmin(f_n)^{1.22})$.
\end{restatable}

Using the lifting theorem of \Goos et al.~\cite{GLMWZ15} (see also \cite[Appendix~A]{Goos15}), \thm{UC1vsC-new} implies the following
\begin{theorem}[$\UPcc(f)$ vs $\coNPcc(f)$] \label{thm:conp1vsupcc-new}
	There exists an infinite family of Boolean functions $f_n: \B^n \times \B^n \to \B$ such that
	$\coNPcc(f_n) \geq \Omega(\UPcc(f_n)^{1.22})$.
\end{theorem}
Hence, the exponent between $\coNPcc$ and $\UPcc$ is somewhere between $1.22$ and $2$. We conjecture the latter to be tight. Moreover, we get as a corollary an improved lower-bound for the conondeterministic communication complexity of the Clique vs Independent Set problem. 

\begin{restatable}{corollary}{cor:CIS}\label{cor:CIS}
There is a family of graphs $G$ such that 
$$
\coNPcc(\mathrm{CIS}_G)\geq \Omega(\log^{1.22}n).
$$

\end{restatable}
We refer the reader to \cite{Goos15} for a discussion on the Clique vs Independent Set problem that shows how \thm{conp1vsupcc-new}  implies \cor{CIS}.

\paragraph{Limitations of \thm{main}.}
We note that $\UCmin(f)$ upper bounds $\deg(f)$,
so this technique cannot be used to get super-quadratic
separations between $\deg(f)$ and $\s(f)$.
A natural question is whether we can use \thm{main}
to get a super-quadratic separation between $\bs(f)$
and $\s(f)$. To do so,
it would suffice to separate $\bs(f)$ from $\UCmin(f)$.
It would even suffice
to separate randomized certificate complexity $\RC(f)$
(a measure larger than $\bs(f)$)
from $\UCmin(f)$, because of the following theorem.

\begin{restatable}[\protect{\cite[Corollary~3.2]{KT16}}]
{theorem}{bsRC}
\label{thm:bsRC}
If there exists a family of functions with $\RC(f) \ge \Omega(\s(f)^{2+\alpha})$, then there exists a family of functions with $\bs(g) \ge \Omega(\s(g)^{2+\alpha-o(1)})$.
\end{restatable}

Unfortunately, we show that separating $\RC(f)$ from
$\UCmin(f)$ is impossible. We conclude that \thm{main}
cannot be used to super-quadratically
separate $\bs(f)$ from $\s(f)$.

\begin{restatable}{theorem}{RCUC}\label{thm:RCUC}
Let $f:\B^n\to\B$ be a Boolean function. Then
$\RC(f)\leq 2\UCmin(f)-1$.
\end{restatable}

We show that the
factor of $2$ in \thm{RCUC} is necessary.
In \app{technical} we strengthen this theorem to show that $\RC(f)$
also lower bounds one-sided conical junta degree.

%
%

\paragraph{Organization.} In \sec{prelim},
we briefly define the many complexity
measures mentioned here, and discuss the known
relationships between them.
{
In \sec{main}, we prove \thm{main} and \cor{main1}. In \sec{main2} we prove \thm{UC1vsC-new}, from which \cor{main2} follows.}
In \sec{failed}, we discuss a failed attempt to get a new
separation between $\bs(f)$ and $\s(f)$, and in the process we
prove \thm{RCUC}.

\section{Preliminaries}\label{sec:prelim}

\subsection{Query Complexity}

Let $f:\B^n \to\B$ be a Boolean function.
Let $A$ be a deterministic
algorithm that computes $f(x)$ on input
$x\in\B^n$ by making queries to the bits of $x$.
The worst-case number of queries $A$ makes (over choices
of $x$) is the query complexity of $A$. The minimum query
complexity of any deterministic algorithm computing $f$
is the deterministic query complexity of $f$, denoted
by $\D(f)$.

We define the bounded-error
randomized (respectively quantum) query complexity
of $f$, denoted by $\R(f)$ (respectively $\Q(f)$),
in an analogous way. We say an algorithm $A$ computes $f$
with bounded error if $\Pr[A(x)=f(x)]\geq 2/3$ for all
$x\in\B^n$, where the probability is over the internal
randomness of $A$. Then $\R(f)$ (respectively $\Q(f)$)
is the minimum number of queries required by any
randomized (respectively quantum) algorithm that computes $f$
with bounded error. It is clear that $\Q(f)\leq\R(f)\leq\D(f)$.
For more details on these measures,
see the survey by Buhrman and de Wolf \cite{BdW02}.

\subsection{Partial Assignments and Certificates}

A partial assignment is a string $p\in\{0,1,*\}^n$
representing partial knowledge of a string $x\in\B^n$.
Two partial assignments are consistent if they agree
on all entries where neither has a $*$.
We will identify $p$ with the set $\{(i,p_i):p_i\neq *\}$.
This allows us to write $p\subseteq x$ to denote that
the string $x$
is consistent with the partial assignment $p$.
We observe that if $p$ and $q$ are consistent partial
assignments, then $p\cup q$ is also a partial assignment.
The size of a partial assignment $p$ is $|p|$,
the number of non-$*$ entries in $p$.
The support of $p$ is the set
$\{i\in [n]: p_i\neq *\}$.

Fix a Boolean function $f:\B^n\to\B$.
We say a partial assignment $p$ is a certificate
(with respect to $f$) if
$f(x)$ is the same for all strings $x\supseteq p$.
If $f(x)=0$ for such strings, we say $p$ is a $0$-certificate;
otherwise, we say $p$ is a $1$-certificate.
We say $p$ is a certificate for the string
$x$ if $p$ is consistent with $x$. We use
$\C_x(f)$ to denote the size of the smallest certificate for
$x$. We then define the certificate complexity of $f$ as
$\C(f):=\max_{x\in\B^n}\C_x(f)$. We also define
the one-sided measures
$\C_0(f):=\max_{x\in f^{-1}(0)}\C_x(f)$
and
$\C_1(f):=\max_{x\in f^{-1}(1)}\C_x(f)$.

\subsection{Sensitivity, Block Sensitivity and Tree Sensitivity}\label{sec:s-bs-ts}

Let $f:\B^n\to\B$ be a Boolean function, and let $x\in\B^n$
be a string. A block is a subset of $[n]$. If $B$ is a block,
we denote by $x^B$ the string we get from $x$ by flipping
the bits in $B$; that is, $x^B_i=x_i$ if $i\notin B$,
and $x^B=1-x_i$ if $i\in B$. For a bit $i$, we also use
$x^i$ to denote $x^{\{i\}}$. 

We say that a block $B$
is sensitive for $x$ (with respect to $f$) if
$f(x^B)\neq f(x)$. We say a bit $i$ is sensitive for $x$
if the block $\{i\}$ is sensitive for $x$.
The maximum number of disjoint blocks
that are all sensitive for $x$
is called the block sensitivity of $x$ (with respect to $f$),
denoted by $\bs_x(f)$. The number of sensitive bits for $x$
is called the sensitivity of $x$,
denoted by $\s_x(f)$. Clearly, $\bs_x(f)\geq\s_x(f)$,
since $\s_x(f)$ has the same definition as $\bs_x(f)$
except the size of the blocks is restricted to $1$.

We now define the measures
$\s(f)$, $\s_0(f)$, and $\s_1(f)$
analogously to $\C(f)$, $\C_0(f)$, and $\C_1(f)$.
That is, $\s(f)$ is the maximum of $\s_x(f)$ over all $x$,
$\s_0(f)$ is the maximum where $x$ ranges over $0$-inputs
to $f$, and $\s_1(f)$ is the maximum over $1$-inputs.
We define $\bs(f)$, $\bs_0(f)$, and $\bs_1(f)$ similarly.

A generalization of sensitivity, defined and studied in \cite{GSTW16}, is \emph{tree sensitivity}. A Boolean function $f:\B^n\to \B$, defines a subgraph $F$ of the the Boolean hypercube $Q_n$, where we keep each edge between $x$ and $x^i$ if $f(x)\neq f(x^i)$. The tree sensitivity of $f$, denoted $\ts(f)$, is  defined to be the size of the largest subtree of $F$, that has at most one edge in each direction $i\in [n]$. Note that sensitivity is the maximum such tree size, when we restrict trees to be stars. In \cite{GSTW16} it was proved that for every Boolean function $f$, $\ts(f)\geq \sqrt{\D(f)}$, and moreover was conjectured that $\sqrt{\D(f)}$ can be replaced by $\D(f)$. 

\begin{conjecture}[\cite{GSTW16}]\label{conjtree}
For every $f$, $\ts(f)\geq \D(f)$. 
\end{conjecture} 

We disprove this conjecture by giving families of functions that have a quadratic separation between tree sensitivity and $\Q(f)$.  This shows that the relation $\ts(f) \ge \sqrt{\D(f)}$ is tight up to poly-logarithmic factors.

\subsection{Fractional Block Sensitivity}

Let $f:\B^n\to\B$ be a Boolean function, and let $x\in\B^n$
be a string. Note that the support of any certificate $p$
of $x$ must have non-empty intersection with every
sensitive block $B$
of $x$; this is because otherwise, $x^B$ would be
consistent with $p$, which is a contradiction since
$f(x^B)\neq f(x)$.

Note further that any subset $S$ of $[n]$ that intersects
with all sensitive blocks of $x$ gives rise to
a certificate $x_S$ for $x$. This is because if $x_S$ was
not a certificate, there would be an input $y\supseteq x_S$
with $f(y)\neq f(x)$. If we write $y=x^B$, where $B$ is
the set of bits where $x$ and $y$ disagree, then $B$
would be a sensitive block that is disjoint from $S$, which
contradicts our assumption on $S$.

This means the certificate complexity $\C_x(f)$ of $x$
is the hitting number for the set system of sensitive
blocks of $x$ (that is, the size of the minimum set
that intersects all the sensitive blocks).
Furthermore, the block sensitivity $\bs_x(f)$ of $x$
is the packing number for the same set system
(i.e. the maximum number of disjoint sets in the system).
It is clear that the hitting number is always larger
than the packing number, because if there are $k$ disjoint
sets we need at least $k$ domain elements in order
to have non-empty intersection with all the sets.

Moreover, we can define the fractional certificate complexity
of $x$ as the fractional hitting number of the
set system; that is,
the minimum amount of non-negative weight we can distribute
among the domain elements $[n]$ so that every set in the
system gets weight at least $1$ (where the weight
of a set is the sum of the weights of its elements).
We can also define the fractional block sensitivity
of $x$ as the fractional packing number of the set
system; that is, the maximum amount of non-negative weight
we can distribute among the sets (blocks) so that
every domain element gets weight at most $1$
(where the weight of a domain element is the sum of the
weights of the sets containing that element).

It is not hard to see that the fractional hitting
and packing numbers are the solutions to dual linear
programs, which means they are equal.
We denote them by $\RC_x(f)$ for ``randomized
certificate complexity'', following the
original notation as introduced by Aaronson \cite{Aar08}
(we warn that our definition differs
by a constant factor from Aaronson's original definition).
We define $\RC(f)$, $\RC_0(f)$, and $\RC_1(f)$
in the usual way.
For more properties of $\RC(f)$, see \cite{Aar08}
and \cite{KT16}.

\subsection{Unambiguous Certificate Complexity}
\label{sec:uc}

Fix $f:\B^n\to\B$. We call a set of partial assignments $U$
an unambiguous collection of $0$-certificates for $f$ if

\begin{enumerate}
    \item Each partial assignment in $U$ is a $0$-certificate
    (with respect to $f$)
    \item For each $x\in f^{-1}(0)$, there is some $p\in U$
    with $p\subseteq x$
    \item No two partial assignments in $U$ are consistent.
\end{enumerate}

We then define $\UC_{0}(f)$ to be the minimum value
of $\max_{p\in U} |p|$ over all choices of
such collections $U$.
We define $\UC_{1}(f)$ analogously, and set
$\UC(f):=\max\{\UC_{0}(f),\UC_{1}(f)\}$.
We also define the one-sided version,
$\UCmin(f):=\min\{\UC_{0}(f),\UC_{1}(f)\}$.

\subsection{Degree Measures}\label{sec:degs}

A polynomial $q$ in the variables
$x_1,x_2,\ldots,x_n$ is said to represent
the function $f:\B^n\to\B$ if
$q(x)=f(x)$ for all $x\in\B^n$.
$q$ is said to $\epsilon$-approximate $f$ if $q(x)\in[0,\epsilon]$
for all $x\in f^{-1}(0)$ and $q(x)\in[1-\epsilon,1]$ for all
$x\in f^{-1}(1)$.
The degree of $f$, denoted by
$\deg(f)$, is the minimum degree of a polynomial representing $f$.
The $\epsilon$-approximate degree, denoted by $\adeg^\epsilon(f)$,
is the minimum degree of a polynomial $\epsilon$-approximating $f$.
We will omit $\epsilon$ when $\epsilon=1/3$.
\cite{BBC+01} showed that $\D(f)\geq\deg(f)$,
$\R(f)\geq\adeg(f)$, and $\Q(f)\geq\adeg(f)/2$.

We also define non-negative variants of degree.
For each partial assignment $p$ we identify a polynomial
$p(x):=\left(\Pi_{i:\,p_i=1} x_i\right)\left(\Pi_{i:\,p_i=0} (1-x_i)
\right)$. We note that $p(x)=1$ if $p\subseteq x$ and
$p(x)=0$ otherwise, and also that the degree of $p(x)$ is $|p|$.
We say a polynomial is non-negative if it is of the form
$\sum_p w_p p(x)$, where $w_p\in\mathbb{R}^+$
are non-negative weights. For such a sum, define its degree
as $\max_{p:\,w_p>0}|p|$. Define its average degree as
the maximum over $x\in\B^n$ of $\sum_{p:\,p\subseteq x} w_p|p|$.
We note that if a non-negative polynomial
$q$ satisfies $|q(x)|\in[0,1]$ for all $x\in\B^n$,
then the average degree of $q$ is at most its degree. Moreover,
if all the monomials in $q$ have the same size and $q(x)=1$
for some $x\in\B^n$, the degree and average degree of $q$
are equal.

We define the non-negative degree of $f$ as the minimum
degree of a non-negative polynomial representing $f$.
We note that this is a one-sided measure, since it may
change when $f$ is negated; we therefore denote it by
$\degp_1(f)$, and use $\degp_0(f)$ for the
degree of a non-negative polynomial representing the negation of
$f$. We let $\degp(f)$ be the maximum of the two, and let
$\degpmin(f)$ be the minimum. We also define $\avdegp_1(f)$
as the minimum average degree of a non-negative polynomial
representing $f$, with the other corresponding measures defined
analogously. Finally, we define the approximate variants
of these, denoted by (for example) $\adegp{,\epsilon}_1(f)$, in a
similar way, except the polynomials need only to
$\epsilon$-approximate $f$.

\subsection{Known Relationships}

\subsubsection{Two-Sided Measures}

We describe some of the known relationships between these measures.
To start with, we have
\[\s(f) \leq \bs(f) \leq\RC(f) \leq\C(f) \leq \UC(f) \leq \D(f),\]
where the last inequality holds because for each deterministic
algorithm $A$, the partial assignments defined by the input
bits $A$ examines when run on some $x\in\B^n$ form an
unambiguous collection of certificates. We also have
\[\adeg(f)\leq 2\Q(f),\quad \adeg^+(f)\leq \R(f),
\quad \deg^+(f)\leq\D(f),\]
with $\adeg(f)\leq\adeg^+(f)\leq\deg^+(f)$ and
$\Q(f)\leq\R(f)\leq\D(f)$.

\cite{BBC+01} showed $\D(f)\leq\bs(f)\C(f)$, and
\cite{Nis91} showed $\C(f)\leq\bs(f)^2$. From this we conclude that
$\D(f)\leq\C(f)^2$ and $\D(f)\leq\bs(f)^3$. \cite{KT16} showed
$\sqrt{\RC(f)}=O(\adeg(f))$; thus
\[\D(f)\leq\bs(f)^3\leq\RC(f)^3=O(\adeg(f)^6)=O(\Q(f)^6),\]
so the above measures are polynomially related (with the exception
of sensitivity). Other known relationships are
$\RC(f)=O(\R(f))$ (due to \cite{Aar08}),
$\D(f)\leq\bs(f)\deg(f)\le \deg(f)^3$ (due to \cite{Mid04}),
and $\deg^+(f)\leq\UC(f)$ (since we can get a polynomial representing
$f$ by summing up the polynomials corresponding to unambiguous
$1$-certificates of $f$).

\subsubsection{One-Sided Measures}

One-sided measures such as $\C_1(f)$ are not polynomially related
to the rest of the measures above, as can be seen from
$\C_1(\OR_n)=1$. This makes them less interesting to us. On the
other hand, the one sided measures $\degpmin(f)$,
$\adegpmin{}(f)$,
and $\UCmin(f)$ are polynomially related to the rest. An easy
way to observe this is to note that $\adegpmin{}(f)\geq\adeg(f)$,
which follows from the fact that $\adeg(f)\leq\adegp{}_1(f)$
and that $\adeg(f)$ is invariant under negating $f$.
Similarly, $\deg(f)\leq\degpmin(f)$. We also have
\[\adegpmin{}(f)\leq\degpmin{}(f)\leq\UCmin(f),\]
where the last inequality holds since we can form a non-negative
polynomial representing $f$ by summing up the polynomials
corresponding to a set of unambiguous $1$-certificates.

An additional useful inequality is $\D(f)\leq \UCmin(f)^2$.
The analogous statement in communication complexity was shown by \cite{Yan91}. The query complexity
version of the proof can be found in \cite{Goos15}.
%

\section{Sensitivity and Unambiguous Certificates}\label{sec:main}

We start by defining a transformation that takes a function
$f$ and modifies it so that $\s_0(f)$ decreases to $1$.
This transformation might cause $\s_1(f)$ to increase,
but we will argue that it will remain upper bounded by
$3\UC_{1}(f)$. We will also argue that other measures, such
as $\D(f)$, do not decrease.
This transformation is motivated by the construction of \cite{Tal16}
that was used to give a power $2.115$ separation between $\D(f)$
and $\s(f)$.

\begin{definition}[Desensitizing Transformation]\label{def:desen}
Let $f:\B^n\to\B$. Let $U$ be an unambiguous collection
of $1$-certificates for $f$, each of size at most $\UC_{1}(f)$.
For each $x\in f^{-1}(1)$, let $p_x\in U$ be the unique certificate
in $U$ consistent with $x$. The desensitized version of
$f$ is the function $f':\B^{3n}\to\B$ defined by
$f'(xyz)=1$ if and only if $f(x)=f(y)=f(z)=1$ and $p_x=p_y=p_z$.
\end{definition}

The following lemma illustrates key properties of $f'$.

\begin{lemma}[Desensitization]\label{lem:desen}
Let $f'$ be the desensitized version of $f:\B^n\to\B$. Then
$\s_0(f')=1$ and $\UC_{1}(f')\leq 3\UC_{1}(f)$.
Also, for any complexity measure
\[M\in\{\D,\R,\Q,\C,\C_0,\C_1,\bs,\bs_0,\bs_1,\RC,
\RC_0,\RC_1,\UC,\UC_{0},\UC_{1},\UCmin,\deg,\degp,\adeg,\adegp{}\},\]
we have $M(f')\geq M(f)$.
\end{lemma}

\begin{proof}
%
We start by upper bounding $\s_0(f')$.
Consider any $0$-input $xyz$ to $f'$
which has at least one sensitive bit. Pick a sensitive bit $i$
of this input; without loss of generality, this bit is inside
the $x$ part of the input. Since flipping $i$ changes $xyz$
to a $1$-input for $f'$, we must have $f(x^i)=f(y)=f(z)=1$ and
$p_{x^i}=p_y=p_z$. In particular, it must hold that $f(y)=f(z)=1$
and $p_y=p_z$. Let $p:=p_y$, so $p=p_z$ and $p=p_{x^i}$.
Since $f(xyz)=0$, it must be the case that $x$ is not consistent
with $p$. Since $p$ is consistent with $x^i$, it must be the case
that $p$ and $x$ disagree exactly on the bit $i$.

Now, it's clear that $xyz$ cannot have any sensitive bits inside
the $y$ part of the input, because then $x$ would not be consistent
with $p_z$. Similarly, $xyz$ cannot have sensitive bits in the $z$
part of the input. Any sensitive bits inside the $x$ part of the
input must make $x$ consistent with $p$; but $x$ disagrees with $p$
on bit $i$, so this must be the only sensitive bit. It follows that
the sensitivity of $xyz$ is at most $1$, as desired. We conclude
that $\s_0(f')=1$.

Next, we upper bound $\UC_{1}$. Define
$U':=\{ppp:p\in U\}\subseteq\{0,1,*\}^{3n}$.
We show that this is an
unambiguous collection of $1$-certificates for $f'$.
First, note that for $p\in U$, if $ppp\subseteq xyz$, then
$f(x)=f(y)=f(z)=1$ and $p_x=p_y=p_z=p$, so
$f'(xyz)=1$. Thus $U'$ is a set of $1$-certificates.
Next, if $xyz$ is a $1$-input for $f'$, then
$f(x)=f(y)=f(z)=1$ and $p_x=p_y=p_z$, which means
$p_xp_xp_x\subseteq xyz$. Since $p_x\in U$, we have
$p_xp_xp_x\in U'$.
Finally, if $ppp,qqq\in U'$ with $ppp\neq qqq$, then
$p\neq q$ and $p,q\in U$, which means $p$ and $q$
are inconsistent. This means $ppp$ and $qqq$ are inconsistent.
This concludes the proof that $U'$ is an
unambiguous collection of $1$-certificates for $f'$.
We have
$\max_{ppp\in U'} |ppp|=3\cdot\max_{p\in U} |p|=3\cdot \UC_{1}(f)$,
so $\UC_{1}(f')\leq 3\cdot \UC_{1}(f)$.

Finally, we show that almost all complexity measures do not decrease
in the transition from $f$ to $f'$. To see this, note that
we can restrict $f'$ to the promise that all inputs come from
the set $\{xyz\in\B^{3n}:x=y=z\}$. Under this promise, the function
$f'$ is simply the function $f$ with each input bit occurring
$3$ times. But tripling input bits in this way does not affect
the usual complexity measures (among the measures
defined in \sec{prelim}, sensitivity is the only exception),
and restricting to a promise can only cause them to decrease.
This means that $f'$ has higher complexity than $f$ under almost
any measure.
\end{proof}

We are now ready to prove Corollary~\ref{cortree}, which disproves Conjecture~\ref{conjtree}.

\cortree*

\begin{proof}
We will use the cheat sheet function $f:=\BKK_\CS$
from \cite{ABK15} that quadratically separates quantum
query complexity from exact degree. This function has $\Q(f) = \tOmega(\UCmin(f)^2)$, as shown
in \cite{AKK16}. Let $f'$ be the desensitized version of $f$, so that by Lemma~\ref{lem:desen}, $\s_0(f')=1$ and $\Q(f') = \tOmega(\s_1(f')^2)$. We finally observe that for every Boolean function with $\s_0(f')=1$, we have $\s(f')=\ts(f')$. This is due to the fact that when $\s_0(f')=1$ every sensitive tree of $f'$ is a star. 
\end{proof}

We now prove \thm{main}, which we restate here for convenience.

\main*

\begin{proof}
Fix $f:\B^n\to\B$ from the family
for which $\D(f)=\tOmega(\UCmin(f)^{1+\alpha})$.
By negating $f$ if necessary, assume
$\UC_{1}(f)=\UCmin(f)$. Apply the desensitizing transformation
to get $f'$. By \lem{desen}, we have $\s_0(f')\leq 1$ and
$\s_1(f')\leq\UC_{1}(f')\leq 3\UCmin(f)$, and also
$\D(f')\geq \D(f)$. We now consider the function
$\hat{f}:=\OR_{3\UCmin(f)}\circ f'$. It is not hard to see
that $\s_0(\hat{f})\leq 3\UCmin(f)$ and
$\s_1(\hat{f})=\s_1(f')\leq 3\UCmin(f)$, so
$\s(\hat{f})\leq 3\UCmin(f)$.

We now analyze $\D(\hat{f})$. We have $\D(f')\geq \D(f)$;
since deterministic query complexity satisfies a perfect
composition theorem, we have
\[\D(\hat{f})=\D(\OR_{3\UCmin(f)})\D(f')\geq 3\UCmin(f)\D(f)
=\tOmega(\UCmin(f)^{2+\alpha})=\tOmega(\s(\hat{f})^{2+\alpha}).\]
This concludes the proof for deterministic query complexity.

For other measures, we need the following properties:
first, that the measure is invariant under negating the function
(so that we can assume $\UCmin(f)=\UC_{1}(f)$ without loss of
generality); second, that the measure satisfies a composition
theorem, at least in the case that the outer function is $\OR$;
and finally, that the measure is large for the $\OR$ function.
We note that the measures $\C$, $\bs$, $\RC$, $\R$, and $\Q$
all satisfy a composition theorem of the form
$M(\OR\circ g)\geq \Omega(M(\OR)M(g))$; for the first three
measures, this can be found in \cite{GSS16}, for $\R$
it can be found in \cite{GJPW15}, and for $\Q$ it follows
from a general composition theorem \cite{Rei11,LMR+11}.
Moreover, $\bs(\OR_n)=\C(\OR_n)=\RC(\OR_n)=n$
and $\R(\OR_n)=\Omega(n)$. This completes the proof for
these measures; for $\Q$, we will have to work harder, since
$\Q(\OR_n)=\Theta(\sqrt{n})$.

For quantum query complexity, the trick will be to use
the function ``Block $k$-sum'' defined in \cite{ABK15}.
It has the property that all inputs have certificates
that use very few $0$ bits. Actually, we'll swap the $0$s
and $1$s so that all inputs have certificates that use very few $1$
bits. When $k=\log n$ (where $n$ the size of the input), we denote
this function by $\BSum_n$. \cite{ABK15} showed that
$\Q(\BSum_n)=\tOmega(n)$, and every input has a certificate with
$O(\log^3 n)$ ones.

Consider the function $\hat{f}:=\BSum_{\UCmin(f)}\circ f'$.
We have
$\Q(\hat{f})=\Q(\BSum_{\UCmin(f)})\Q(f')=\tOmega(\UCmin(f)\Q(f))$.
We now analyze the sensitivity of $\hat{f}$. Fix an input
$z$ to $\hat{f}=\BSum_{\UCmin(f)}\circ f'$. This input
consists of $\UCmin(f)$ inputs to $f'$, which, when evaluated,
form an input $y$ to $\BSum_{\UCmin(f)}$. Note that some of the
inputs to $f'$ correspond to sensitive bits of $y$ (with
respect to $\BSum_{\UCmin(f)}$); the sensitive bits of $z$
are  then simply the sensitive bits of those inputs.
Now, consider the certificate of $y$ that uses only
$O(\log^3 \UCmin(f))$ bits that are $1$. Since it is a certificate,
it must contain all the sensitive bits of $y$; thus at most
$O(\log^3 \UCmin(f))$ of the $1$ bits of $y$ are sensitive.
It follows that the number of sensitive bits of $z$ is at most
$\UCmin(f)\s_0(f')+O(\log^3\UCmin(f))\s_1(f')
=\tO(\UCmin(f))$. This concludes the proof.
\end{proof}

It is not hard to see that the same approach can yield separations
against bounded-size block sensitivity (where the blocks
are restricted to have size at most $k$). To do this,
we need the desensitizing construction to repeat the inputs
$2k+1$ times instead of $3$ times. Instead of increasing
to $3\UCmin(f)$, the bounded-size block sensitivity
would increase to $(2k+1)\UCmin(f)$, and the deterministic
query complexity would increase to $(2k+1)\D(f)$. When $k$
is constant, we get the same asymptotic separations as for
sensitivity.

We now construct separations against $\UCmin$. This proves
\cor{main1} and \cor{main2}.

\cormaina*
\begin{proof}
By \thm{main}, it suffices to construct a family of functions with
 $\Q(f)=\tOmega(\UCmin(f)^2)$. Our function will be a cheat sheet function $\BKK_\CS$
from \cite{ABK15} that quadratically separates quantum
query complexity from exact degree. This function has quantum
query complexity quadratically larger than $\UCmin$, as shown
in \cite{AKK16}.
\end{proof}

\cormainb*

\begin{proof}
In \thm{UC1vsC-new}, we construct a family of functions with
$\C(f)=\tOmega(\UCmin(f)^{{\frac{\log(38/3)}{\log(8)}}})$.
Thus, by \thm{main}, we can construct a family of functions with
$\C(f)=\tOmega(\s(f)^{1+{\frac{\log(38/3)}{\log(8)}}}) = \Omega(\s(f)^{2.22})$.
\end{proof}



\section{Improved separation between $\UC_1$ and $\CC$}\label{sec:main2}
In this section we prove \thm{UC1vsC-new},
building on the proof by \cite{Goos15}. Our main contribution is to show how to adapt the argument in \cite{Goos15} to allow for fractional weights. We finally give a fractional weighting scheme that leads to our improved separation.  
We observe that in order to obtain our final result, one can just take \Goos's construction and reweight it in the end. Nonetheless, we include the full details here to show that any gadget with a separation between $\UC_1$ and $\C$ implies an asymptotic separation (which was not explicit in \cite{Goos15}).

Throughout the section, $\Sigma$ and $\Gamma$ will denote finite sets that correspond to  input and output alphabets of our functions. We shall assume that $0$ is not in $\Sigma$, and will discuss functions $f: (\{0\} \cup \Sigma)^n \to \Gamma$ where $0$ is a special symbol treated differently than others.

\subsection{Certificates and Weighted-Certificates for Large-Alphabet Functions}

We generalize the definition of certificates from Boolean functions to functions with arbitrary input and output alphabets.

\begin{definition}[Multi-valued Certificates, Simple Certificates]
A {\sf certificate} for a function $f: (\{0\} \cup \Sigma)^n \to \Gamma$ is a cartesian product of sets $S_1 \times S_2 \times \ldots S_n$ where each $S_i \subseteq \{0\} \cup \Sigma$ is a non-empty set and such that all $y \in S_1 \times S_2 \times \ldots S_n$ have the same $f$-value.

A {\sf simple certificate} for $f$ is a certificate where each $S_i$ is either: (i) $\{0\} \cup\Sigma$, or (ii) $S_i$ contains exactly one element, and this element is from $\Sigma$ (i.e., not the $0$ element).%
\footnote{Note that a certificate for ``$f(x)=1$'' for a Boolean function $f: \B^n \to \B$ is always simple.}

We define the {\sf size} of a certificate as the number of $i$'s such that $S_i \neq (\{0\} \cup \Sigma)$.
For $x \in (\{0\} \cup \Sigma)^n$, we denote by $\CC(f, x)$ the  size of the smallest certificate for $f$ which contains $x$.

For a set $T\subseteq \Gamma$ we say that $S_1 \times \ldots \times S_n$ certifies that ``$f(\cdot) \in T$'' if this is true for any $y \in S_1 \times \ldots \times S_n$. When $T = \Gamma \setminus \{i\}$ we write ``$f(\cdot) \neq i$'' for shorthand.
\end{definition}

\begin{definition}[Weight Schemes, Certificate Weights]
\label{def:weight-scheme}
Let $w:\Sigma \to \RR^{+}$ be a non-negative weight function.
A {\sf weight scheme} is a mapping, $w$, associating positive real numbers to non-empty subsets of $\{0\} \cup \Sigma$ such that:
\begin{enumerate}	
\item If $S=\{i\}$, for some $i \in \Sigma$, then the weight of $S$ is $w(i)$.
\item If $S = (\{0\} \cup \Sigma)$, then the weight of $S$ is $0$.
\item \label{item:weight S with 0} If $0 \in S$ an $S \neq  (\{0\} \cup \Sigma)$, then the weight of $S$ is $\max_{i \in \Sigma \setminus S}\{w(i)\}$.
(In particular, if $S=\{0\}$ then the weight of $S$ is $\max_{i \in \Sigma}\{w(i)\}$.)
\end{enumerate}
(Note that we did not specify the weight of sets $S$ of at least two elements which do not contain $0$, as they will not be used in our analysis.)

The weight of a certificate $S_1 \times \ldots \times S_n$ is simply $\sum_{i=1}^{n}{w(S_i)}$.
For a function $f: (\{0\} \cup \Sigma)^n \to \Gamma$ and an input $x\in(\{0\} \cup \Sigma)^n$ we define the certificate complexity $\CC((f,w),x)$ to be the minimal weight of a certificate $S_1 \times \ldots \times S_n$ for $f$ according to $w$, such that $x \in S_1 \times \ldots \times S_n$.
\end{definition}

\begin{definition}[Realization of Weight Schemes]
	The weight-scheme defined by an {\sf integer-valued} weight function $w: \Sigma \to \N$ is {\sf realized} by $g_w: (\{0\} \cup \Sigma)^{m} \to (\{0\} \cup \Sigma)$ if:
\begin{description}
\item[(i)] For $i\in \Sigma$, there exists a collection of unambiguous  certificates of size-$(w(i))$ for  ``$g_w(\cdot)=i$'',
\item[(ii)] $g_w(0^m) = 0$, and  
\item[(iii)] In order to prove ``$g_w(0^m) \in S$'' it is required to expose at least $w(S)$ coordinates of $0^m$.
\end{description}
\end{definition}

\begin{lemma}[Weight-Scheme Implementation, \cite{Goos15}]
\label{lem:weight-implementation}Let  $w: \Sigma \to \N$ be an integer-valued weight function. Then, there exists a weight scheme associating natural numbers to non-empty subsets of $\{0\} \cup \Sigma$ that can be realized by a function $g_w:(\{0\} \cup \Sigma)^m \to (\{0\} \cup \Sigma)$ where $m = \max_{i}\{w_i\}$.
\end{lemma}
\begin{proof}
We define $g_w(x) = i$ iff the symbol $i$ appears in the first $w(i)$ coordinates and $i$ is the first non-zero symbol to appear in the string.
We set $g_w(x) = 0$ if there is no such $i\in \Sigma$.
(ii) holds trivially.
For (i) note that the decision tree that queries the first $w(i)$ coordinates induces an unambiguous collection of certificates for ``$g_w(\cdot)=i$''.
For (iii) we may assume without loss of generality that $S \neq (\{0\} \cup \Sigma)$ as otherwise the claim is trivial. Since we are proving that ``$g_w(0^m) \in S$'' and indeed $g_w(0^m)=0$ it is required that $0\in S$.
It remains to show that it is required to expose the first $\max_{i\in \Sigma \setminus S}{w(i)}$ coordinates of the input to $g_w$. 
Let $i$ be the element in $\Sigma \setminus S$ with maximal weight.
Indeed, if one coordinate in the first $w(i)$ coordinates was not exposed, then it is still possible that $g_w(\cdot) = i$, as all coordinates that were exposed are equal to $0$ and there is an unexposed position in the first $w(i)$ coordinates that might be marked with $i$. 
\end{proof}

\subsection{Composing Functions over Large Alphabet with Fractional Weights} \label{sec:composition}
Most of the results below are generalizations of arguments from \cite{Goos15}. However, since unlike \cite{Goos15} we deal with fractional weights, in addition to the total weight, we also need to take into account the number of coordinates in the intermediate certificates. 

Let $f:(\{0\} \cup \Sigma)^N \to \{0,1\}$, and let $w : \Sigma \to \RR^{+}$.
We treat the pair $(f,w)$ as a ``weighted function''.
Let $\mathcal{C}$ be an unambiguous collection of simple $1$-certificates of size-$s$ and weight at most $W$ for $(f,w)$. 
Let $\Sigma_0$ be a finite set that does not contain $0$ and  $w_0 : \Sigma_0 \to \RR^{+}$. We define $(\tilde{f}, \tilde{w})$,  where $\tilde{f}: (\{0\} \cup \Sigma \times \Sigma_0)^{N} \to (\{0\} \cup \Sigma_0)$ as follows. Denote by $\pi_1(x)$ and $\pi_2(x)$ the projection of $x\in (\{0\}\cup \Sigma \times \Sigma_0)^N$ to its $(\{0\}\cup \Sigma)^N$ coordinates and its $(\{0\}\cup \Sigma_0)^N$ coordinates respectively. The value of $\tilde{f}(x)$ is defined as follows. 

\begin{quote}
	If $f(\pi_1(x)) = 0$, then set $\tilde{f}(x) :=0$. Otherwise, let $\mathcal{T} \in \mathcal{C}$ be the unique certificate for ``$f(\cdot)=1$'' on $\pi_1(x)$. Read the corresponding coordinates of $\mathcal{T}$ from $\pi_2(x)$ and if all of them are equal to some $i\in \Sigma_0$, then set $\tilde{f}(x) :=i$; otherwise set $\tilde{f}(x) :=0$. 
\end{quote}
Let $\tilde{w}: \Sigma \times \Sigma_0 \to \RR^{+}$ be defined as $\tilde{w}(\sigma, i) =  w(\sigma) \cdot w_0(i)$. The following lemma shows useful bounds on the certificates of the new function $\tilde{f}$ according to $\tilde{w}$.

\begin{lemma}[From Boolean to Larger Output Alphabet]\label{lem:larger output}
Let $\tilde{f}$, $f$, $\tilde{w}$,  $w$ and $w_0$ be defined as above. Then,
\begin{description}
	\item[(B1)] There is an unambiguous collection of simple size-$s$ certificates for ``$\tilde{f}(\cdot) = i$'' with weight at most $w_0(i) \cdot W$ according to $\tilde{w}$.
	\item[(B2)] The certificate complexity of ``$\tilde{f}(0^N) \neq i$'' with respect to $\tilde{w}$ is at least $w_0(i) \cdot \CC((f,w),0^N)$.
	\end{description}
\end{lemma}

\begin{proof}
	\begin{description}
	\item[(B1)] The unambiguous collection of simple $1$-certificates for $f$ corresponds to unambiguous collection of simple $i$-certificates for $\tilde{f}$ by checking that each queried symbol has its $\Sigma_0$-part equals $i$. The weight of each certificate in the collection is at most $w_0(i) \cdot W$ as each coordinate weighs $w_0(i)$ times its ``original'' weight.

		\item[(B2)] Fix $i\in \Sigma_0$. Assume we have a certificate for ``$\tilde{f}(0^n) \neq i$''. This is a cartesian product $S_1 \times \ldots \times S_N$ such that each $S_i$ contains the $0$ symbol and under which $\forall{x\in S_1 \times \ldots \times S_N}$ it holds that $\tilde{f}(x) \neq  i$.
		Take $\hat{f}$ to be $\tilde{f}$ restricted	only to input alphabet $\{0\} \cup (\Sigma \times \{i\})$.
		Then $S'_1 \times \ldots \times S'_n$ where $S'_j = S_j \cap (\{0\} \cup (\Sigma \times \{i\}))$ is a certificate for ``$\hat{f}(0^N)\neq i$''.  Using property~\ref{item:weight S with 0} in Definition~\ref{def:weight-scheme}, we show that $w(S'_j) \le w(S_j)$.
		We consider two cases. If $S'_j = \{0\} \cup (\Sigma \times \{i\})$, then $w(S'_j) = 0 \le w(S_j)$.
		Otherwise, $$w(S'_j) = \max_{\sigma\in (\Sigma \times \{i\})\setminus S'_j} {w(\sigma)} \le \max_{\sigma \in (\Sigma \times \Sigma_0) \setminus  S_j}{w(\sigma)} = w(S_j).$$

		However, proving that ``$\hat{f}(0^N) = 0$'' is equivalent to proving that ``$f(0^N) = 0$'', except for the reweighting. Since each coordinate weighs according to $\tilde{w}$ at least $w_0(i)$ times its weight according to $w$, the weight of the certificate $S_1 \times \ldots \times S_n$ is at least $w_0(i) \cdot C((f,w),0^N)$.\qedhere
		\end{description}
\end{proof}
%

\begin{lemma}[Composition Lemma]\label{lem:composition}
	Let $h: (\{0\} \cup [k])^{n} \to \{0,1\}$ with $h(0^n) = 0$, and let $w_0:[k] \to \RR^{+}$ such that $(h,w_0)$ has an unambiguous collection of simple $1$-certificates of size-$k$ and (fractional) weight $u$, however any certificate for ``$h(0^n)=0$'' is of (fractional) weight $v$.
	
	Let $\tilde{f}$ and $\tilde{w}$ be as defined above. Let $f' : (\{0\} \cup \Sigma \times [k])^{n \times N} \to \{0,1\}$ be defined by $f' = h \circ \tilde{f}$, and $w': (\{0\} \cup \Sigma \times [k]) \to \N$ be equal to $\tilde{w}$. Then,
	\begin{description}
	\item[(A1)] $(f',w')$ has an unambiguous collection of simple certificates $1$-certificates with size at most $sk$ and weight at most $u \cdot W$.
	\item[(A2)] 
	$\CC((f',w'),0^{Nn}) \ge v \cdot \CC((f,w), 0^N)$.
\end{description}
\end{lemma}

\begin{proof}
	\begin{description}
	\item[(A1)] Take the unambiguous collection $\mathcal{C}$ of simple $1$-certificates for $h$ of size-$k$ and (fractional) weight $u$.
	For any certificate $\mathcal{T}$ from $\mathcal{C}$ replace the verification that some coordinate equals $i$ with the simple certificate that the relevant $N$-length input of $\tilde{f}$ belongs to $\tilde{f}^{-1}(i)$. The cost of each such certificate to $\tilde{f}$ will be at most $W \cdot w_0(i)$ according to $\tilde{w} \equiv w'$.
	Thus, the overall cost will be $W \cdot u$, and the certificates will be of size at most $sk$. It is easy to verify that these certificates are unambiguous, since unambiguous collections of simple certificates are closed under composition.
	
	\item[(A2)] 
		Let $\mathcal{T}$ be a certificate for ``$f'(0^{N\cdot n})=0$'' of minimal weight (according to $w'$), and let $w_{\mathcal{T}}$ be its weight.
		Let $\mathcal{T}_1, \ldots, \mathcal{T}_n$ be the substrings of $\mathcal{T}$ of length $N$ according to the composition of $h \circ \tilde{f}$.
		By Lemma~\ref{lem:larger output}[B1], if $\mathcal{T}_i$ certifies that ``$\tilde{f}(0^N) \neq j$'', then it costs at least $w_0(j) \cdot \CC((f,w),0^N)$.	
		We construct a certificate $\mathcal{H}$ for $h$ from $\mathcal{T}$.
		If $\mathcal{T}_i$ certifies that $\tilde{f}(0^N) \neq j$ then $(\mathcal{H})_i \neq j$. More formally, let $\mathcal{H} = S_1 \times \ldots \times S_n$, where for $i\in[n]$ the set $S_i$ consists of $\{0\}$ union with all $j$ such that $\mathcal{T}_i$ {\em does not} certify that $\tilde{f}(0^N) \neq j$.
		Suppose by contradiction that $\mathcal{H}$ does not certify that ``$h(0^n)=0$''. Then, there exists an input $y \in S_1 \times \ldots \times S_n$ (i.e., an input consistent with $\mathcal{H}$) such that $h(y)=1$. Thus, there exist inputs $x^{(1)}, \ldots, x^{(n)}$ each of length $N$ such that $\tilde{f}(x^{(i)}) = y_i$ and $\mathcal{T}_i$ is consistent with $x^{(i)}$, which shows that $\mathcal{T}$ is not a certificate for $h \circ \tilde{f}_{\tilde{w}}$.
		Thus, $\mathcal{H}$ is a certificate for $h(0^n)=0$, and we get that 
		\[w_\mathcal{T} \ge w_0(\mathcal{H}) \cdot \CC((f,w),0^N) = v \cdot \CC((f,w),0^N).\qedhere\]
	\end{description}
\end{proof}

Next, we show how to take any ``gadget'' $h$ -- a function over a constant number of symbols -- with some gap between the $\UC_1(h)$ and $\C(h,0^n)$, and convert it into an infinite family of functions with a polynomial separation between $\UC_1$ and $\C$.

\begin{theorem}[From Gadgets to Boolean Unweighted Separations]\label{thm:gadget to separation}
Let $u, v \in \RR$, $k\in \N$ be constants such that $1 \le k \le u < v$.
Let $h: (\{0\} \cup [k])^{n} \to \{0,1\}$ with $h(0^n) = 0$, and let $w_0:[k] \to \RR^{+}$ such that $(h,w_0)$ has an unambiguous collection of simple $1$-certificates of size-$k$ and (fractional) weight $u$, however any certificate for ``$h(0^n)=0$'' is of (fractional) weight $v$.

Then, there exists an infinite family of Boolean functions $\{h'_m\}_{m\in \N}$ with
\begin{enumerate}
	\item $\UC_1(h'_m) \le \poly(m) \cdot u^m$
	\item $\CC(h'_m) \ge v^m$
	\item $h'_m$ is defined over $\poly(m)\cdot \exp(O(m))$ many bits.
\end{enumerate}
\end{theorem}
\begin{proof}
We start by defining a sequence of weighted functions $\{(h_m,w_m)\}_{m\in \N}$ over large alphabet size with a polynomial gap between $\UC_1$ and $\CC$. We then convert these functions into unweighted Boolean functions with the desired properties.

We take $h_1 := h$ and $w_1 := w_0$.
For $m\ge 2$ we take $(h_m,w_m)$ to be the composition of $(h,w_0)$ with $(\tilde{h}_{m-1},\tilde{w}_{m-1})$.
Let $\Sigma_m = [k]^m$. Then, $h_m : (\{0\} \cup \Sigma_m)^{n^m} \to \{0,1\}$ and $w_m : (\{0\} \cup \Sigma_m) \to \RR^{+}$.
Using Lemma~\ref{lem:composition}, we have that 
\begin{description}
	\item[(i)] The maximal weight in $w_{m}$ is at most $(w_{0,max})^m$, where $w_{0,max}:=  \max_i\{w_0(i)\}$.
	\item[(ii)] \label{item:UC composed} 
	There exists an unambiguous collection of simple $1$-certificates of size $k^m$ and weight at most $u^m$ for $(h_m,w_m)$.
	\item[(iii)] $\CC((h_m,w_m),\vec{0}) \ge v^m$.
\end{description}

\paragraph*{Making Weights Integral.}
First, we modify the weights so that they will be integral. 
We take $w'_m(\cdot)$ to be $\lceil{w_m(\cdot)\rceil}$.
Taking ceiling on the weights may only increase the certificate complexities. Thus, $\CC((h_m,w'_m),\vec{0}) \ge v^m$.
On the other hand, the weight of any certificate may only increase additively by its size, hence $\UC_1((h_m, w'_m)) \le u^m + k^m \le 2u^m$.

\paragraph*{Eliminating Weights.}
Next, we convert the weighted function $(h_m, w'_m)$ to an unweighted Boolean function $h'_m$ with similar $\UC_1$ and $\CC$ complexities.
First, we remove the weights by applying Lemma~\ref{lem:weight-implementation} (using the fact that $w'_m$ is integer-valued). We define $h''_m  =  h_m \circ g_{w'_m}$.
Lemma~\ref{lem:weight-implementation}  implies that
\begin{align*}
	\CC(h''_m) \ge \CC((h_{m},w'_m)) \ge v^m
\end{align*}
and 
\begin{align*}
	\UC_1(h''_m) \le \UC_1((h_{m},w'_m)) \le 2\cdot u^m.
\end{align*}

\paragraph*{Booleanizing.}
To make the inputs of the function $h''_m$ Boolean we repeat 
the argument of \Goos \cite{Goos15}.
If $f$ is a function $f:\Sigma^{N} \to \B$, we may always convert it to a boolean function by composing it with some surjection $g_{\Sigma}: \{0,1\}^{\lceil{\log |\Sigma|\rceil}} \to \Sigma$. The following naive bounds will suffice for our purposes:
\begin{equation}\label{eq:Booleanizing}
	\mathcal{C}(f) \le \mathcal{C}(f \circ g_{\Sigma}) \le \mathcal{C}(f) \cdot \lceil{\log |\Sigma|\rceil}
\quad\text{forall~} \mathcal{C}\in \{\UC_1, \CC\}.
\end{equation}
In our final alphabet $\Sigma = \{0\} \cup [k]^m$, thus $h'_m = h''_m \circ g_{\Sigma}$ is a Boolean function with 
$$\CC(h'_m) \ge \CC(h''_m) \ge v^m$$
and 
$$\UC_1(h'_m) \le \UC_1(h''_m) \cdot \lceil{\log |\Sigma|\rceil}  \le 2 \cdot  u^m \cdot O(m\log k).$$

\paragraph*{Input Length.}
The input length of $h_m$ is  $n^m$. By lemma~\ref{lem:weight-implementation}, the input length of $h''_m$ is at most $n^m \cdot (w_{0,max}^m+1)$.
Thus the input length to $h'_m$ is at most
\[O(\log(|\Sigma|) \cdot (n\cdot w_{0,max})^m) = O( m \cdot \log( k)  \cdot (n\cdot w_{0,max})^m)\;\qedhere\]
\end{proof}

\subsection{Gadgets Based on Projective Planes.}\label{sec:projectiveplane}
We will use a reweighted version of the function constructed by \Goos~\cite{Goos15} based on projective planes as our gadget. Let us first recall the definition of a projective plane.

\begin{definition}[Projective plane]
A projective plane is a $k$-uniform hyper-graph with $n=k^2-k+1$ edges and $n$ nodes with the following properties.
\begin{itemize}
\item Each node is incident on exactly $k$ edges.
\item For every two nodes, there exists a unique edge containing both. 
\item Every two edges intersect on exactly one node.
\end{itemize}
\end{definition}
Given a projective plane, it follows from Hall's theorem that it is possible to assign an ordering to the edges incident to each vertex in a way that for each edge, its assigned order for each of its nodes is different. Namely, for each $i$, there are no two nodes for which their $i$-th incident edge is the same. 

It is well-known that projective planes exist for every $k$ such that $k-1$ is a prime power. \Goos~\cite{Goos15} introduced the following function $f: (\{0\}\cup \Sigma)^n \to \bool$ based on a projective plane, with $\Sigma= [k]$.  We think of the inputs of $f$ as a sequence of pointers, one for each node, where $0$ is the Null pointer, and $i\in [k]$ is a pointer to the $i$-th edge on which the node is incident on. We set $f(x)=1$ if there is an edge of the projective plane such that all its nodes point to it, and $f(x)=0$ otherwise. 

We will be interested in showing a gap between the certificate complexity of ``$f(0^n)=0$'' and $\UC_1(f)$. However, the function as is, allows a certificate of size $k$ for ``$f(0^n)=0$'' matching its $\UC_1(f)$. One certificate for ``$f(0^n)=0$'' is to pick an arbitrary edge of the projective plane, and certify that all its nodes have the Null pointer. This certifies ``$f(0^n)=0$'' as every two edges in a projective plane intersect on a node. An unambiguous collection of size $k$ certificates consists of picking for each edge all its nodes and ensuring that they point to that edge. This collection is unambiguous using the same property that every two edges intersect on one node. 

In order to obtain a gadget with a gap between $\UC_1$ and $\CC$, \Goos introduced weights on the input alphabet of $f$. Each element $i \in \Sigma$ is assigned a weight $w(i)$, where the weights are intended to carry the following meaning: For every $i\in \Sigma$ it costs $w(i)$ for a certificate to assure that ``$x_j=i$'', and moreover $0$ has the special property that it costs $\max_{i\in \Sigma} w(i)$ to assure that ``$x_j=0$'' (as in Definition~\ref{def:weight-scheme}).  In \cite{Goos15} each $i\in [k]$ is assigned a weight $w(i)=i$.  \Goos~\cite{Goos15} implemented this weighting scheme specifically for the case when $w(i):=i$ via a weighting gadget $g_w:(\{0\} \cup \Sigma)^k \to (\{0\} \cup \Sigma)$ (as done in \lem{weight-implementation}) and considering $f\circ g_w$. Our improvement comes from considering a different weighting scheme with fractional weights. 


\begin{claim}[Reweighting the Projective Plane]\label{claim:newweights}
Let $f$ be defined as above, and let $w(i):=\frac{i}{{(k+1)/2}}$. Then, $(f,w)$ has an unambiguous collection of simple $1$-certificates of size $k$ and weight $k$. Moreover, any certificate for $f(0^n)=0$ is of weight at least $\frac{{k^2-k+1}}{(k+1)/2}$
\end{claim}
\begin{proof}
 \Goos \cite[Claims~6~and~7]{Goos15} showed  that with respect to the weight-function $w'(i) = i$, the function $f$ has an unambiguous collection of simple $1$-certificates of size-$k$ and weight $(k\cdot (k+1))/2$. However, any certificate for ``$f(0^n)=0$'' is of weight at least ${k^2 - k+1}$.

From this, it is immediate that with respect to $w\equiv \frac{w'}{(k+1)/2}$, $f$ has an unambiguous collection of simple $1$-certificates of size-$k$ and (fractional) weight $\frac{(k\cdot (k+1))/2}{(k+1)/2} = k $. However, any certificate for $0^n$ is of weight at least $\frac{k^2 - k+1}{(k+1)/2}$. 
\end{proof}

\subsection{Putting Things Together}
Given a gadget $(h,w_0)$ such that $h$ has unambiguous collection of simple $1$-certificates of size-$k$ and (fractional) weight $u$, however any certificate for $0^n$ is of (fractional) weight $v$, with $v>u>1$ and $u \ge k$, Theorem~\ref{thm:gadget to separation} gives a polynomial separation between $\CC$ and $\UC_1$:
\begin{equation}
\CC(h'_m) \ge v^m = (u^m)^{\log(v)/\log(u)} \ge \widetilde{\Omega}\(\UC_1(h'_m)^{\log(v)/\log(u)}\)\;.\label{eq:UCvsC}
\end{equation}
We take $h$ to be the projective plane function $f$ described in Section~\ref{sec:projectiveplane} with $k=8$, $n = k^2 - k + 1=57$ and weight function $w_0(i) = \frac{i}{(k+1)/2}$.  By Claim~\ref{claim:newweights}, we have that with respect to $w_0$, $h$ has an unambiguous collection of simple $1$-certificates of size-$k$ and weight $k = 8$. However, any certificate for $0^n$ is of  weight $\frac{k^2 - k+1}{(k+1)/2} = 38/3$.
Plugging these values in Equation~\eqref{eq:UCvsC} we get a better separation:
\begin{equation}
	\CC(h'_m) \ge \widetilde{\Omega}
	\(\UC_1(h'_m)^{\frac{\log(38/3)}{\log(8)}}\)
	\ge \Omega(\UC_1(h'_m)^{1.22})\;,
\end{equation}
where the input length is $N \le \poly(m) \cdot \exp(O(m))$.
The lifting theorem  of \cite{GLMWZ15,Goos15} incurs a loss factor of $\log(N) = O(m)$ in the separation, however this is negligible compared to the $\poly(m) \cdot u^m$ versus $v^m$ separation.

\subsection{Further Improvements}
Since our theorem is general in transforming a fractional weighted gadget into a polynomial separation, it is enough to only improve the gadget construction in order to improve the $\UC_1$ vs $\CC$ exponent.
Indeed, even using the same gadget (the projective plane function of \Goos) we can consider different weight function.
Using computer search it seems that such reweighting is indeed better than our choice of $w_0$. However, the improvement is mild and currently we do not have a humanly verifiable proof for the lower bound on the certificate complexity of $0^n$ under the reweighting. Indeed, \Goos relied on the fact that the weights were $w'(i)=i$ in order to present a simple proof of his lower bound on the certificate complexity of ``$h(0^n) = 0$'' according to $w'$.
It seems though (we have verified this using computer-search for small values of $k$) that the best weights are attained by taking $w'(i) = i+1$ and then reweighting by multiplying all weights by the constant $\alpha = \frac{1}{(k+3)/2}$, so that the unambiguous certificates for $h$ will be of weights $k$. We leave proving a lower bound under this weight function as an open problem.

\section{Attempting a Super-Quadratic Separation vs.\ Block Sensitivity}\label{sec:failed}

In this section, we describe why attempting to use
\thm{main} to get a super-quadratic separation between $\bs(f)$
and $\s(f)$ fails. In the process, we show some new lower bounds
for $\UCmin(f)$ and even for the one-sided non-negative
degree measures.

One approach for the desired super-quadratic separation is to
find a family of functions for which $\bs(f)\gg \UCmin(f)$.
In fact, by \cite{KT16}, it suffices to provide a family of functions for which $\RC(f)\gg \UCmin(f)$ (as explained in \sec{RCsep}). In \sec{RCrelations}, we show that even separating
$\RC(f)$ from $\UCmin(f)$ is impossible: we have
$\RC(f)\leq 2\UCmin(f) -1$. This means our techniques do not give
anything new for this problem. This is perhaps surprising,
since $\RC(f)$ is similar to $\C(f)$, yet
\cite{Goos15} showed a separation between $\C(f)$ and $\UCmin(f)$.

\subsection{A Separation Against \texorpdfstring{$\RC(f)$}{Lg}
is Sufficient}\label{sec:RCsep}

\cite{KT16} showed that a separation between $\s(f)$ and $\RC(f)$
implies an equal separation between $\s(f)$ and $\bs(f)$ (see
 \thm{bsRC}).
The key insight is that $\bs(f)$ becomes $\RC(f)$ when the function
is composed enough times; this was observed by \cite{Tal13}
and by \cite{GSS16}. This means that if we start
with a function separating $\s(f)$ and $\RC(f)$ and compose
it enough times, we should get a function with the same
separation between $\s(f)$ and $\RC(f)$, but with the additional
property that $\bs(f) \approx \RC(f)$.

\subsection{But \texorpdfstring{$\RC(f)$}{Lg} Lower Bounds
\texorpdfstring{$\UCmin(f)$}{Lg}}\label{sec:RCrelations}

We would get a super-quadratic separation between $\bs(f)$
and $\s(f)$ if we had a super-linear separation between $\RC(f)$
and $\UCmin(f)$. Unfortunately, this is impossible using our paradigm, as we now show.
Actually, we can prove an even stronger statement, namely that
$\RC(f)\leq (2\aavdegpmin{,\epsilon}(f)-1)/(1-4\epsilon)$.
We note that this implies \thm{RCUC},
because when $\epsilon=0$, we have
\[\RC(f)\;\leq\; 2\avdegpmin(f)-1\;\leq\; 2\degpmin(f)-1\;\leq\; 2\UCmin(f)-1.\]
This stronger statement says that one-sided conical junta degree is lower bounded
by two-sided randomized certificate complexity, which helps clarify the hierarchy
of lower bounds for randomized algorithms.

The proof of the relationship
$\RC(f)\leq (2\aavdegpmin{,\epsilon}(f)-1)/(1-4\epsilon)$
is somewhat technical; we leave it for \app{technical}, and provide
a cleaner proof (of $\RC(f)\leq 2\UCmin(f)-1$) below.
One interesting thing to note about it
is that it holds for partial functions as well, as long as
the definition of $\aavdegpmin{,\epsilon}(f)$ requires
the approximating polynomial to evaluate to at most $1$
on the entire Boolean hypercube.

Before providing the proof, we'll provide a warm up proof that
$\bs(f)\leq 2\UCmin(f)$.

\begin{lemma}\label{lem:bsUC}
For all non-constant $f:\B^n\to\B$,
we have $\bs(f)\leq 2\UCmin(f)-1$.
\end{lemma}

\begin{proof}
Without loss of generality,
we have $\UCmin(f)=\UC_{1}(f)$. We also have
$\bs_1(f)\leq\C_1(f)\leq\UC_{1}(f)$, so it remains to show that
$\bs_0(f)\leq 2\UC_{1}(f)-1$. Also without loss of generality,
we assume that the block sensitivity of $0^n$ is
$\bs(f)$ and that $f(0^n)=0$.

Let $B_1,B_2,\ldots,B_{\bs(f)}$
be disjoint sensitive blocks of $0^n$. Let $U$ be an unambiguous
collection of $1$-certificates for $f$, each of size at most
$\UC_{1}(f)$. For each $i\in[\bs(f)]$, we have $f(\vec{0}^{B_i})=1$,
so there is some $1$-certificate $p_i\in U$ such that $p_i$ is
consistent with $\vec{0}^{B_i}$. Since $p_i$ is a $1$-certificate,
it is not consistent with $\vec{0}$, so it has a $1$ bit
(which must have index in $B_i$). Now, if $i\neq j$, the certificate
$p_i$ has a $1$ inside $B_i$ and only $0$ or $*$ symbols outside
$B_i$, and the certificate $B_j$ has a $1$ inside $B_j$
and only $0$ or $*$ symbols outside $B_j$; thus $p_i$ and $p_j$
are different. Since $U$ is an unambiguous collection, $p_i$
and $p_j$ must conflict on some bit (with one of them assigning $0$
and the other assigning $1$),
or else there would be an input consistent with both.

We construct a directed graph on vertex set $[\bs(f)]$ as follows.
For each $i,j\in[\bs(f)]$ with $i\neq j$,
we draw an arc from $i$ to $j$
if $p_i$ has a $0$ bit in a location where $p_j$ has a $1$ bit.
It follows that for each pair $i,j\in[\bs(f)]$ with $i\neq j$,
we either have an arc from $i$ to $j$
or else we have an arc from $j$ to $i$ (or both).
The number of arcs in this graph is at least
$\bs(f)(\bs(f)-1)/2$, so the average out degree is at least
$(\bs(f)-1)/2$. Hence there is some vertex $i$ with out degree
at least $(\bs(f)-1)/2$. But this means $p_i$ conflicts with
$(\bs(f)-1)/2$ other certificates
$p_{j_1},p_{j_2},\ldots,p_{j_{(\bs(f)-1)/2}}$ with $p_i$
having a bit $0$ and $p_{j_k}$ having a $1$-bit; however,
two different certificates $p_{j_x}$ and $p_{j_y}$ cannot
both agree on a $1$ bit, since the $1$ bits of $p_{j_x}$
must come from block $B_{j_x}$ and the blocks are disjoint.
This means $p_i$ has at least $(\bs(f)-1)/2$ zero bits.
It must also have at least one $1$ bit. Thus
$|p_i|\geq \bs(f)/2+1/2$, so $\bs(f)\leq 2\UCmin(f)-1$.
\end{proof}

We now generalize this lemma from $\bs$ to $\RC$, proving \thm{RCUC}. A further
strengthening of the result can be found in \app{technical}.

\RCUC*

\begin{proof}
Without loss of generality,
we have $\UCmin(f)=\UC_{1}(f)$. We also have
$\RC_1(f)\leq\C_1(f)\leq\UC_{1}(f)$, so it remains to show that
$\RC_0(f)\leq 2\UC_{1}(f)-1$. Also without loss of generality,
we assume that the randomized certificate of $0^n$ is
$\RC(f)$ and that $f(0^n)=0$.

We prove the theorem using the characterization of $\RC(f)$ as the fractional block sensitivity of $f$.
Let $B_1,B_2,\ldots,B_{m}$ be minimal sensitive blocks of $0^n$.
Let $a_1, \ldots, a_m$ be weights assigned to blocks $B_1, \ldots, B_m$ such that 
\[
 \sum_{j}{a_j} = \RC(f)\quad \text{, and}\quad \forall{i \in [n]} :\sum_{j: i \in B_j} {a_j} \le 1\;.
\]
Let $U$ be an unambiguous collection of $1$-certificates for $f$, each of size at most $\UC_{1}(f)$. For each $i\in[m]$, we have $f(\vec{0}^{B_i})=1$,
so there is some $1$-certificate $p_i\in U$ such that $p_i$ is
consistent with $\vec{0}^{B_i}$. Since $p_i$ is a $1$-certificate,
it is not consistent with $\vec{0}$, so it has a $1$ bit
(which must have index in $B_i$). 
Next, we show that if $i\neq j$, then $p_i$ and $p_j$ are different.
Assume by contradiction that $p_i = p_j$, then $p_i$ is a partial assignment that satisfy both $\vec{0}^{B_i}$ and $\vec{0}^{B_j}$, hence it must satisfy $\vec{0}^{B_i\cap B_j}$, but this means that $f(\vec{0}^{B_i \cap B_j}) = 1$ which contradicts the fact that both $B_i$ and $B_j$ are minimal sensitive blocks for $\vec{0}$.

We established that for any $i\neq j$, 
the partial assignments $p_i$ and $p_j$ are different.
Since $U$ is an unambiguous collection, $p_i$
and $p_j$ must conflict on some bit (with one of them assigning $0$
and the other assigning $1$),
or else there would be an input consistent with both.

We construct a directed weighted graph on vertex set $[m]$ as follows.
For each $i,j\in[m]$ with $i\neq j$,
we draw an arc from $i$ to $j$ with weight $a_{i} \cdot a_{j}$,
if $p_i$ has a $0$ bit in a location where $p_j$ has a $1$ bit.
It follows that for each pair $i,j\in[m]$ with $i\neq j$,
we either have an arc from $i$ to $j$
or else we have an arc from $j$ to $i$ (or both).
The total weight of the arcs in this graph is
\begin{align*}
\sum_{i<j}{a_{i} \cdot a_{j} \cdot (|p_i^{-1}(1) \cap p_j^{-1}(0)| + |p_i^{-1}(0) \cap p_j^{-1}(1)|)} &\ge\;
\sum_{i<j}{a_{i} \cdot a_{j}}\\
&=
\frac{1}{2} \cdot (\sum_{i}{a_{i}})^2 - \frac{1}{2} \cdot \sum_{i}{a_{i}^2}  \\
&\ge\; \frac{1}{2} \cdot (\sum_{i}{a_{i}})^2 - \frac{1}{2} \cdot \sum_{i}{a_{i}} \tag{$a_i \le 1$}\\
&\ge\; \frac{1}{2} \cdot (\RC(f)^2 - \RC(f))
\end{align*}
Note that by symmetry, the LHS equals
$$\sum_{i} a_{i} \cdot \sum_{j\neq i} {a_j \cdot |p_i^{-1}(0) \cap p_j^{-1}(1)|}.$$
Since $\sum_{i}{a_i} = \RC(f)$, by averaging,
\begin{equation}\label{eq:exists i}
\exists{i}:\;
 \frac{1}{2} (\RC(f)-1) \le \sum_{j\neq i} {a_j \cdot |p_i^{-1}(0) \cap p_j^{-1}(1)|}.
\end{equation}  
Next, we get a lower bound on $|p_i^{-1}(0)|$ from Eq.~\eqref{eq:exists i}.
\begin{align*}
	 \frac{1}{2} (\RC(f)-1) &\le \sum_{j\neq i} {a_j \cdot |p_i^{-1}(0) \cap p_j^{-1}(1)|}\\
	 &= \sum_{k: p_i(k) = 0}\;\;  \sum_{j:p_j(k) = 1} {a_j}\\
	 &\le \sum_{k: p_i(k) = 0} \;\; \sum_{j:k \in B_j} {a_j}\tag{$p_j$ is consistent with $\vec{0}^{B_j}$}\\
	 &\le |p_i^{-1}(0)|. \tag{$\sum_{j:k \in B_j} {a_j}\le 1$ for all $k$}
\end{align*}
We showed that $p_i$ has at least $(\RC(f)-1)/2$ zero bits.
It must also have at least one $1$ bit. Thus
$|p_i|\geq \RC(f)/2+1/2$, so $\RC(f)\leq 2\UCmin(f)-1$.
\end{proof}

We note that the relationships in \lem{bsUC} and \thm{RCUC} are tight.\footnote{We thank Mika G\"o\"os
	for helping to simplify this construction.}
Let $k$ be any non-negative integer, we construct a function $f$ on $n = 2k+1$ variables with $\s(f) = \bs(f) = \RC(f) = n$ and $\UCmin(f) \le k+1$. 
This shows that the inequalities $\bs(f) \le 2\UCmin(f) -1 $ and $\RC(f) \le 2\UCmin(f)-1$ are both tight for all values of $\UCmin(f)$.
We define the function $f$ by describing a set of partial assignments $p_0, \ldots, p_{n-1}$  such that $f(x)=1$ if and only if $\exists{i}: p_i \subseteq x$.
 Let  $p = 0^k 1 {*}^k$. The assignments $p_0, \ldots, p_{n-1}$ are all possible cyclic-shifts of $p$, namely for $0 \le i \le k$,
 $p_i = 0^{k-i} 1 {*}^{k} 0^i$ and for $k+1\le i\le 2k$ we have
 $p_i = {*}^{2k + 1 - i} 0^{k} 1 {*}^{i-1 -k}$. It is easy to verify that any two different partial assignments $p_i$ and $p_j$ are not consistent with one another. Hence, $p_0, \ldots, p_{n-1}$ is an unambiguous collection of $1$-certificates for $f$, each of size $k+1$, exhibiting that $\UCmin(f) \le k+1$.
On the other hand, $f(0) = 0$ and for all $i\in [n]$, we have $f(e_i) = 1$, showing that $f$ has sensitivity $n$ on the all-zeros input. Overall, we showed that $\s(f) = \bs(f) = \RC(f) = n = 2k+1$ while $\UCmin(f) \le k$.

\section*{Acknowledgements}

We would like to thank Mika G\"o\"os and Robin Kothari for many helpful discussions and for comments on a preliminary draft.
We also thank the anonymous referees of ITCS for their comments.

\appendix

\section{Lower Bound for Approximate Non-Negative Degree}
\label{app:technical}

Here we show that the lower bound in \thm{RCUC} holds even for one-sided
average approximate non-negative degree, the smallest version of conical
junta degree. This is saying that conical juntas, in all their forms, give a more
powerful lower bound technique for randomized algorithms than $\RC(f)$.

\begin{theorem}\label{thm:avdeg}
	Let $f:\B^n\to\B$ be a non-constant function,
	and let $\aavdegpmin{,\epsilon}(f)$ denote the
	minimum average degree of a non-negative polynomial that approximates
	either $f$ or its negation with error at most $\epsilon$
	(see \sec{degs} for definitions).
	If $\epsilon<1/4$, we have
	$\RC(f)\leq \frac{2\aavdegpmin{,\epsilon}(f)-1}{1-4\epsilon}$.
\end{theorem}

\begin{proof}
	Let $q$ be the non-negative approximating polynomial with
	average degree $\aavdegpmin{,\epsilon}(f)$.
	Without loss of generality,
	we assume $q$ approximates $f$ rather than its negation.
	We can write $q\equiv \sum_{p\in\{0,1,*\}} w_p p$,
	so for any $x\in\B^n$, we have
	\[q(x)=\sum_{p\in\{0,1,*\}}w_p p(x)=\sum_{p:\,p\subseteq x} w_p,\]
	where recall that $w_p$ are non-negative weights given to
	partial assignments.
	This means for all $x\in\B^n$, we know that
	\[\left|f(x)-\sum_{p:\,p\subseteq x} w_p\right|\leq \epsilon,
	\qquad \sum_{p:\,p\subseteq x} w_p\leq 1,\qquad
	\mbox{and}\qquad \sum_{p:\,p\subseteq x} w_p|p|
	\leq \aavdegpmin{,\epsilon}(f).\]
	
	Now, consider the input $y\in\B^n$ for which $\RC_y(f)=\RC(f)$.
	There are two cases: either $y$ is a $0$-input, or else
	$y$ is a $1$-input. If $y$ is a $1$-input, we use the fractional
	certificate complexity interpretation of $\RC_y(f)$: the value
	$\RC_y(f)$ is the minimum amount of weight that can be distributed
	to the bits of $y$ such that every sensitive block of $y$
	contains bits of total weight at least $1$. We assign to bit
	$i\in[n]$ the weight
	\[\frac{1}{1-2\epsilon}\sum_{p:\,p\subseteq y, p_i\neq *} w_p.\]
	Then each sensitive block $B\subseteq [n]$ for $y$ satisfies
	$f(y^B)=0$, so the sum of $w_p$ over all $p\subseteq y$ that
	have support disjoint from $B$ must be at most $\epsilon$.
	Since the sum of $w_p$ over all $p\subseteq y$ is at least
	$1-\epsilon$, there must be weight at least $1-2\epsilon$
	assigned to partial assignments consistent with $p$ whose
	support overlaps $B$. It follows that the total weight given
	to the bits in $B$ is at least $1$, which means this weighting
	is feasible. This means the total weight upper bounds $\RC_y(f)$, so
	\[\RC(f)=\RC_y(f)\leq\frac{1}{1-2\epsilon}
	\sum_{i\in[n]\;}\sum_{p:\,p\subseteq y,\,p_i\neq *} w_p
	=\frac{1}{1-2\epsilon}\sum_{p:\,p\subseteq y} w_p|p|
	\leq \frac{\aavdegpmin{,\epsilon}(f)}{1-2\epsilon}.\]
	
	It remains to deal with the case where $y$ is a $0$-input.
	In this case, we use the fractional block sensitivity interpretation
	of $\RC_y(f)$: the value of $\RC_y(f)$ is the maximum amount of
	weight that can be distributed to the sensitive blocks of $y$ such
	that every bit of $y$ lies inside blocks of total weight
	at most $1$. Without loss of generality,
	we can assume only minimal sensitive blocks are assigned weight
	(minimal sensitive blocks are sensitive blocks such that
	all their proper subsets are not minimal).
	
	Let $\mathcal{B}:=\{B\subseteq[n]:f(y^B)\neq f(y)\}$ be the
	set of sensitive blocks of $y$, and let
	$\mathcal{M}:=\{B\in\mathcal{B}: \forall B'\subset B, B'\notin \mathcal{B}\}$
	be the set of minimal sensitive blocks of $y$.
	Let  $\{a_B\}_{B\in\mathcal{M}}$ with $a_B\in\mathbb{R}^+$
	be the optimal weighting of the minimal sensitive blocks.
	This means
	$\sum_{B\in\mathcal{M}} a_B=\RC_y(f)$ and
	$\sum_{B\ni i} a_B\leq 1$ for all $i\in [n]$.
	
	We have $\sum_{p\subseteq y} w_p\leq \epsilon$ and
	$\sum_{p\subseteq y^B} w_p\geq 1-\epsilon$
	for all $B\in\mathcal{B}$. 
	Thus, for any $B_1,B_2\in\mathcal{M}$ with $B_1\neq B_2$, we can write
	\[2-2\epsilon
	\leq \sum_{p\subseteq y^{B_1}} w_p + \sum_{p\subseteq y^{B_2}} w_p
	= \sum_{p\subseteq y^{B_1}:\,p\nsubseteq y^{B_1\cup B_2}} w_p
	+ \sum_{p\subseteq y^{B_2}:\,p\nsubseteq y^{B_1\cup B_2}} w_p
	+ \sum_{p\in G} w_p + \sum_{p\in H} w_p,\]
	where $G:=\{p:p\subseteq y^{B_1},\,p\subseteq y^{B_1\cup B_2}\}$
	and
	$H:=\{p:p\subseteq y^{B_2},\,p\subseteq y^{B_1\cup B_2}\}$.
	The last two sums are equal to
	$\sum_{p\in G\cup H} w_p + \sum_{p\in G\cap H} w_p$.
	We have
	$\sum_{p\in G\cup H} w_p\leq \sum_{p\subseteq y^{B_1\cup B_2}} 
	w_p\leq 1$. Also, any $p\in G\cap H$ satisfies
	$p\subseteq y^{B_1\cap B_2}$. Since $B_1\neq B_2$ and they are
	both minimal sensitive blocks, we have $f(y^{B_1\cap B_2})=0$,
	so $\sum_{G\cap H} w_p\leq \sum_{p\subseteq y^{B_1\cap B_2}} w_p
	\leq \epsilon$. It follows that
	\[\sum_{p\subseteq y^{B_1}:\,p\nsubseteq y^{B_1\cup B_2}} w_p
	+ \sum_{p\subseteq y^{B_2}:\,p\nsubseteq y^{B_1\cup B_2}} w_p
	\geq 1-3\epsilon.\]
	Note that the above sums are over disjoint sets, since
	if $p\subseteq y^{B_1}$ and $p\nsubseteq y^{B_1\cup B_2}$,
	then $p$ must disagree with $y^{B_2}$ on some bit inside $B_2$.
	If we split out the parts of the sums for which $p\subseteq y$,
	we get
	\[\sum_{p\subseteq y} w_p+
	\sum_{p\subseteq y^{B_1}:\,p\nsubseteq y,
		\,p\nsubseteq y^{B_1\cup B_2}} w_p
	+ \sum_{p\subseteq y^{B_2}:\,p\nsubseteq y,
		\,p\nsubseteq y^{B_1\cup B_2}} w_p
	\geq 1-3\epsilon.\]
	Since $f(y)=0$, the first sum is at most $\epsilon$, so
	\[\sum_{p\subseteq y^{B_1}:\,p\nsubseteq y,
		\,p\nsubseteq y^{B_1\cup B_2}} w_p
	+ \sum_{p\subseteq y^{B_2}:\,p\nsubseteq y,
		\,p\nsubseteq y^{B_1\cup B_2}} w_p
	\geq 1-4\epsilon.\]
	
	
	We now write the following.
	\begin{align*}
	\RC(f)^2-\RC(f)
	&=\sum_{B_1\in\mathcal{M}}a_{B_1}\sum_{B_2\in\mathcal{M}}a_{B_2}
	-\sum_{B_1\in\mathcal{M}}a_{B_1}\\
	&\leq \sum_{B_1\in\mathcal{M}}a_{B_1}\sum_{B_2\in\mathcal{M}}a_{B_2}
	-\sum_{B_1\in\mathcal{M}}a_{B_1}^2\\
	&=\sum_{B_1\in\mathcal{M}}a_{B_1}\sum_{B_2\neq B_1}a_{B_2}\\
	&\leq\frac{1}{1-4\epsilon}
	\sum_{B_1\in\mathcal{M}}a_{B_1}\sum_{B_2\neq B_1}a_{B_2}
	\left(\sum_{p\subseteq y^{B_1}:\,p\nsubseteq y,
		\,p\nsubseteq y^{B_1\cup B_2}} w_p
	+ \sum_{p\subseteq y^{B_2}:\,p\nsubseteq y,
		\,p\nsubseteq y^{B_1\cup B_2}} w_p \right)\\
	&=\frac{2}{1-4\epsilon}
	\sum_{B_1\in\mathcal{M}}a_{B_1}\sum_{B_2\neq B_1}a_{B_2}
	\sum_{p\subseteq y^{B_1}:\,p\nsubseteq y,
		\,p\nsubseteq y^{B_1\cup B_2}} w_p,
	\end{align*}
	where the second line follows because $a_{B_1}\leq 1$
	for all $B_1\in\mathcal{M}$.
	
	Note that $\sum_{B_1\in\mathcal{M}} a_{B_1}=\RC(f)$,
	so if we divide both sides by $\RC(f)$, the last line
	becomes a weighted average. It follows that there exists
	some minimal block $B_1$ such that
	\begin{align*}
	\RC(f)-1 &\leq \frac{2}{1-4\epsilon}
	\sum_{B_2\neq B_1}a_{B_2}
	\sum_{p\subseteq y^{B_1}:\,p\nsubseteq y,
		\,p\nsubseteq y^{B_1\cup B_2}} w_p\\
	&= \frac{2}{1-4\epsilon}
	\sum_{p\subseteq y^{B_1}:\,p\nsubseteq y} w_p
	\sum_{B_2\neq B_1:p\nsubseteq y^{B_1\cup B_2}} a_{B_2}.
	\end{align*}
	Examine the inner summation above.
	Note that $y^{B_1\cup B_2}=(y^{B_1})^{B_2\setminus B_1}$.
	Since $p\subseteq y^{B_1}$, the condition
	$p\nsubseteq y^{B_1\cup B_2}$ is equivalent to the support
	of $p$ having non-empty intersection with $B_2\setminus B_1$.
	Using $\supp(p)$ to denote the support of $p$, we have
	\begin{align*}
	\RC(f)-1 &\leq \frac{2}{1-4\epsilon} 
	\sum_{p\subseteq y^{B_1}:\,p\nsubseteq y} w_p
	\sum_{i\in \supp(p)\setminus B_1}
	\sum_{B_2\in\mathcal{M}:\,i\in B_2} a_{B_2}\\
	&\leq \frac{2}{1-4\epsilon} 
	\sum_{p\subseteq y^{B_1}:\,p\nsubseteq y} w_p
	\sum_{i\in \supp(p)\setminus B_1} 1\\
	&= \frac{2}{1-4\epsilon} 
	\sum_{p\subseteq y^{B_1}:\,p\nsubseteq y} w_p
	|\supp(p)\setminus B_1|\\
	&\leq \frac{2}{1-4\epsilon}
	\sum_{p\subseteq y^{B_1}:\,p\nsubseteq y} w_p (|p|-1)\\
	&\leq \frac{2}{1-4\epsilon}\aavdegpmin{,\epsilon}(f)-
	\frac{2}{1-4\epsilon}
	\sum_{p\subseteq y^{B_1}:p\nsubseteq y} w_p\\
	&\leq \frac{2}{1-4\epsilon}\aavdegpmin{,\epsilon}(f)-
	\frac{2}{1-4\epsilon}\left(
	\sum_{p\subseteq y^{B_1}} w_p-\sum_{p\subseteq y} w_p
	\right)\\
	&\leq \frac{2}{1-4\epsilon}\aavdegpmin{,\epsilon}(f)-
	\frac{2}{1-4\epsilon}(1-\epsilon-\epsilon)\\
	&\leq \frac{2}{1-4\epsilon}\aavdegpmin{,\epsilon}(f)-
	\frac{2-4\epsilon}{1-4\epsilon},
	\end{align*}
	where the second line follows because the sum of $a_B$ over
	all blocks $B\in\mathcal{M}$ containing a given element
	$i\in[n]$ is at most $1$, and the fourth line follows
	because the conditions $p\subseteq y^{B_1}$ and
	$p\nsubseteq y$ imply that the support of $p$ is not disjoint
	from $B_1$. Finally, we get
	\[\RC(f)\leq \frac{2}{1-4\epsilon}\aavdegpmin{,\epsilon}(f)
	-\frac{1}{1-4\epsilon}
	=\frac{2\aavdegpmin{,\epsilon}(f)-1}{1-4\epsilon},\]
	as desired.
\end{proof}

\bibliographystyle{alphaurl}
\newcommand{\eprint}[1]{{\small \upshape \tt \href{http://arxiv.org/abs/#1}{#1}}}
\let\oldpath\path
\renewcommand{\path}[1]{\small\oldpath{#1}}
\bibliography{sensitivity}
\end{document}